\documentclass[11pt,a4paper,notitlepage]{article}

\catcode`\|=12\relax

\usepackage{amsmath}
\usepackage{amsfonts}
\usepackage{amsthm}
\usepackage{enumitem}
\usepackage{bm}
\usepackage{geometry,mathtools}
\usepackage{braket}
\usepackage{bbold}
\usepackage{tikz}

\newcounter{braid}
\newcounter{strands}
\pgfkeyssetvalue{/tikz/braid height}{1cm}
\pgfkeyssetvalue{/tikz/braid width}{1cm}
\pgfkeyssetvalue{/tikz/braid start}{(0,0)}
\pgfkeyssetvalue{/tikz/braid colour}{black}
\pgfkeys{/tikz/strands/.code={\setcounter{strands}{#1}}}

\makeatletter
\def\cross{%
  \@ifnextchar^{\message{Got sup}\cross@sup}{\cross@sub}}

\def\cross@sup^#1_#2{\render@cross{#2}{#1}}

\def\cross@sub_#1{\@ifnextchar^{\cross@@sub{#1}}{\render@cross{#1}{1}}}

\def\cross@@sub#1^#2{\render@cross{#1}{#2}}

\def\render@cross#1#2{
  \def\strand{#1}
  \def\crossing{#2}
  \pgfmathsetmacro{\cross@y}{-\value{braid}*\braid@h}
  \pgfmathtruncatemacro{\nextstrand}{#1+1}
  \foreach \thread in {1,...,\value{strands}}
  {
    \pgfmathsetmacro{\strand@x}{\thread * \braid@w}
    \ifnum\thread=\strand
    \pgfmathsetmacro{\over@x}{\strand * \braid@w + .5*(1 - \crossing) * \braid@w}
    \pgfmathsetmacro{\under@x}{\strand * \braid@w + .5*(1 + \crossing) * \braid@w}
    \draw[braid] \pgfkeysvalueof{/tikz/braid start} +(\under@x pt,\cross@y pt) to[out=-90,in=90] +(\over@x pt,\cross@y pt -\braid@h);
    \draw[braid] \pgfkeysvalueof{/tikz/braid start} +(\over@x pt,\cross@y pt) to[out=-90,in=90] +(\under@x pt,\cross@y pt -\braid@h);
    \else
    \ifnum\thread=\nextstrand
    \else
     \draw[braid] \pgfkeysvalueof{/tikz/braid start} ++(\strand@x pt,\cross@y pt) -- ++(0,-\braid@h);
    \fi
   \fi
  }
  \stepcounter{braid}
}

\tikzset{braid/.style={double=\pgfkeysvalueof{/tikz/braid colour},double distance=1pt,line width=2pt,white}}

\newcommand{\braid}[2][]{%
  \begingroup
  \pgfkeys{/tikz/strands=2}
  \tikzset{#1}
  \pgfkeysgetvalue{/tikz/braid width}{\braid@w}
  \pgfkeysgetvalue{/tikz/braid height}{\braid@h}
  \setcounter{braid}{0}
  \let\sigma=\cross
  #2
  \endgroup
}
\makeatother

\newtheorem{theorem}{Theorem}[section]
\newtheorem{proposition}[theorem]{Proposition}%

\newtheorem{definition}{Definition}%

\numberwithin{definition}{subsection}
\numberwithin{theorem}{subsection}
\numberwithin{figure}{subsection} 

\title{\vspace{-1cm}Universal Braiding Quantum Gates}
\author{David Lovitz\footnote{Lovitz@pdx.edu}}

\providecommand{\keywords}[1]
{
  \small	
  \textbf{\textit{Keywords---}} #1
}

\begin{document}

\makeatletter 
\def\@acknow{}%
\long\def\EarlyAcknow#1 \par{%
\def\@acknow{\abstractfont\abstracthead*{Acknowledgments}
#1\par}}%

\def\printabstract{\ifx\@acknow\empty\else\@acknow\fi\par%
    \ifx\@abstract\empty\else\@abstract\fi\par}
\makeatother



\newcommand\blfootnote[1]{%
  \begingroup
  \renewcommand\thefootnote{}\footnote{#1}%
  \addtocounter{footnote}{-1}%
  \endgroup
}

\maketitle

\begin{abstract}
The Yang-Baxter equation and it's various forms have applications in many fields, including statistical mechanics, knot theory, and quantum information. Unitary solutions of the braided Yang-Baxter equation are of particular interest as quantum gates for topological quantum computers. We demonstrate a simple construction for solutions in any dimension, which are both unitary and universal for quantum computation. We also fully classify a family of solutions to certain generalized Yang-Baxter equations and prove that certain instances of the equation only have solutions that are scalar multiples of the identity.
\end{abstract}

\keywords{
Yang–Baxter equation, quantum computing, braid group, topology}

\blfootnote{The author would like to thank Dr. Steven Bleiler for his valuable feedback on this manuscript.}

\section{The Yang-Baxter Equation and Quantum Gates}
A topological quantum computer is a theoretical machine that manipulates topological phases of matter to perform computation. Topological phases of matter have been observed experimentally as the fractional quantum Hall effect \cite{avron_topological_2003}. In particular, quasi-particles called anyons, which arise when electrons are exposed to a magnetic field, have been observed to exhibit a nontrivial phase change when exchanged. The phase change depends only on the number of particles exchanged and not on the specific path taken \cite{nakamura_direct_2020}. Therefore, the changes in phase of a quantum system consisting of $n$ ordered anyons corresponds to the $n$ strand braid group in topology, denoted by $B_n$. What follows is a brief introduction to this group, from which the Yang-Baxter equation naturally arises from the group's representation theory. The Yang-Baxter equation is a matrix equation, and it's solutions, when unitary, correspond to the quantum gates in a topological quantum computer \cite{kauffman_braiding_2004}. 

\subsection{The Braid Group}
The \textit{braid group} $B_n$ was first introduced by E. Artin in 1925 \cite{artin_theorie_1925}. The group $B_n$ can be visualized as $n$ vertical strands which are braided over and under one another, while maintaining that each vertical strand must pass the horizontal line test. That is, each strand must strictly travel from top to bottom without any local maxima or minima. Braids are considered equivalent if they are ambient isotopic with fixed endpoints. Some ambient isotopies give rise to the relators. Multiplication between two braids is defined by placing the braids vertically above each other, and gluing them together as indicated in Figure \ref{braidmult}.

\begin{figure}[!ht]
    \centering
    \begin{tikzpicture}
    \braid[braid colour=black,strands=3,braid start={(0, 0)}]%
    {\sigma_1}
    
    \braid[braid colour=black,strands=3,braid start={(0, 1.3)}]%
    {\sigma_2}
   
    \node[font=\Huge] at (4.5,0) {\(=\)};
     \braid[braid colour=black,strands=3,braid start={(5,0)}]%
    {\sigma_1}
    \braid[braid colour=black,strands=3,braid start={(5,1)}]%
    {\sigma_2}
    \end{tikzpicture}
    \caption{The multiplication of two 3 strand braids.}
    \label{braidmult}
\end{figure}
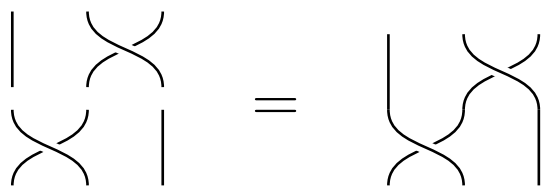

A set of generators for the braid group consists of the set $\{I, \sigma_1, \ldots \sigma_{n-1} \}$ where $\sigma_i$ is the braid with strand $i$ crossing over and to the right of strand $i+1$, and $I$ is the braid with no crossings. There are two relations in the braid group which we will refer to as \textit{braid relations}. The first relation requires that two generators commute as long as they are at least two strands apart: $\sigma_i \sigma_{j}= \sigma_j \sigma_i$ whenever $|i-j|>1$, these relations are sometimes referred to as \textit{far commutativity}. The generators must also satisfy the Yang-Baxter relations: 

\begin{align*} 
\sigma_i \sigma_{i+1} \sigma_i &= \sigma_{i+1} \sigma_i \sigma_{i+1}\\ 
\end{align*}

This relation can be seen in Figure \ref{braidRelation} which shows two braids which are ambient isotopic with fixed endpoints.

\begin{figure}[!ht]
    \centering
\begin{tikzpicture}
\braid[braid colour=black,strands=3,braid start={(0,0)}]%
{\sigma_1 \sigma_2 \sigma_1^{-1}}
\node[font=\Huge] at (4.5,-1.5) {\(=\)};
\braid[strands=3,braid start={(5,0)}]
{\sigma_2 \sigma_1 \sigma_2}
\end{tikzpicture}
\caption{The braid on the right is obtained by ambient isotoping the strands of the braid on the left.}
\label{braidRelation}
\end{figure}
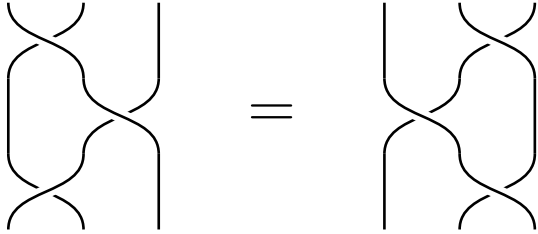

A \textit{representation} of the braid group is a group homomorphism into a group of invertible matrices. If $V$ is an $n$ dimensional vector space over a field $\mathbb{F}$, denote by $V^{\otimes n}$ the tensor product of $V$ with itself $n$ times. It is then natural to look for representations of $B_n$ in $V^{\otimes n}$. One way to define such a braid group representation is by mapping:
\begin{equation}
\label{eq:braidrepBYBE}
\sigma_i \rightarrow I^{\otimes i-1} \otimes R \otimes I^{\otimes n-i-1}
\end{equation}


where $R: V\otimes V \rightarrow V\otimes V$ is an invertible linear map on the tensor product of $V$ with itself. We also let $\otimes$ denote the $aB$ convention Kronecker product between two linear maps. That is the $ij$ block of $A\otimes B$ is given by $(a_{ij} B)$ and $A^{\otimes k}$ denotes the Kronecker product of the matrix $A$ with itself $k$ times. This defines a representation of the braid group as long as all of the braid relations are satisfied.

\subsection{The Braided Yang-Baxter Equation}
\label{sec:bybe}
\begin{definition} Let $V$ be a $d$-dimensional vector space over $\mathbb{C}$. Let $I$ be the $d\times d$ identity matrix on the vector space $V$, and $R: V\otimes V \rightarrow V\otimes V$ an invertible linear transformation. The matrix $R$ satisfies the \textit{braided Yang-Baxter equation} (bYBE) when:
\begin{equation}
\label{braidedYBE}
    (R\otimes I)(I \otimes R)(R\otimes I)=(I\otimes R)(R \otimes I)(I\otimes R)
\end{equation} 
\end{definition}
This is the braided Yang-Baxter equation and any matrix satisfying this equation is sometimes referred to as an R-matrix. The matrix $R$ in the representation defined in equation \ref{eq:braidrepBYBE} must satisfy the matrix form of the braided Yang-Baxter equation defined here, or in references \cite{kassel_quantum_1995} and \cite{birman_braids_1974}.

\subsection{Generalized Yang-Baxter Equation}

\begin{definition} Let $d$, $m$, and $l$ be natural numbers. Let $V$ be a vector space over $\mathbb{C}$ of dimension $d$, and $R:V^{\otimes m} \rightarrow V^{\otimes m}$ be an invertible matrix. Denote the identity on $V$ by $I_V$. The matrix $R$ is a solution to the $(d,m,l)$-\textit{generalized Yang-Baxter equation} whenever

\begin{equation}
\label{gYBE}
    (R\otimes I_V^{\otimes l})(I_V^{\otimes l} \otimes R)(R\otimes I_V^{\otimes l})=(I_V^{\otimes l} \otimes R)(R \otimes I_V^{\otimes l})(I_V^{\otimes l} \otimes R)
\end{equation}
\end{definition}
This gives another way to represent the braid group was introduced by Rowell, Zhang, Wu, and Ge in \cite{rowell_extraspecial_2010} via the so called generalized Yang-Baxter equations (gYBE). When $m=2$ and $l=1$ this expression is equivalent to the bYBE. Any bYBE solution gives rise to a braid group representation while a solution to a gYBE gives rise to a braid group representation whenever the far commutativity relations are also satisfied, which is guaranteed when $l>m/2$ \cite{chen_generalized_2012}. The representation in that case is given by the homomorphism:

\begin{equation}
\sigma_i \rightarrow I_{d^l}^{\otimes i-1} \otimes R \otimes I_{d^l}^{\otimes n-i-1}
\end{equation} 
where $I_{d^l}$ is the $d^l \times d^l$ identity matrix. Finding matrices which satisfy either the bYBE or a gYBE is a difficult task and only a few examples are known, even in low dimensions. The full classification has been completed only for the (2,2,1)-gYBE (equivalent to the 2-dimensional bYBE) in \cite{hietarinta_solving_1993} and the unitary solutions in \cite{dye_unitary_2003}. Unitary solutions are particularly desirable because they appear in topological quantum computation as braiding quantum gates.

\begin{theorem}
All invertible solutions to the $(d,m,l)$-gYBE are of the form $\lambda I_{d^m}$ whenever $l \geq m$.
\end{theorem}
\begin{proof}
Fix $(d,m,l)$ with $l\geq m$, note that this makes $R$ a $d^m\times d^m$ matrix, and $I_V$ the $d\times d$ identity matrix. We can then write:
\begin{align*}
    (R\otimes I_d^{\otimes l})(I_d^{\otimes l} \otimes R)(R\otimes I_d^{\otimes l})&=(I_d^{\otimes l} \otimes R)(R \otimes I_d^{\otimes l})(I_d^{\otimes l} \otimes R)\\
    (R\otimes I_d^{\otimes l})(I_d^{\otimes m}\otimes I_d^{\otimes l-m} \otimes R)(R\otimes I_d^{\otimes l})&=(I_d^{\otimes m}\otimes I_d^{\otimes l-m} \otimes R)(R\otimes I_d^{\otimes l})(I_d^{\otimes l} \otimes R)\\
    (R\otimes I_d^{\otimes l})(I_d^{\otimes m}R \otimes (I_d^{\otimes l-m} \otimes R)I_d^{\otimes l})&=(I_d^{\otimes m}R \otimes (I_d^{\otimes l-m} \otimes R)I_d^{\otimes l})(I_d^{\otimes l} \otimes R)\\
    (R\otimes I_d^{\otimes l})(R \otimes I_d^{\otimes l-m} \otimes R)&=(R \otimes I_d^{\otimes l-m} \otimes R)(I_d^{\otimes l} \otimes R)\\
    (R^2 \otimes I_d^{\otimes l} (I_d^{\otimes l-m} \otimes R))&=((R \otimes I_d^{\otimes l-m})I_d^{\otimes l} \otimes R^2)\\
    (R^2 \otimes I_d^{\otimes l-m} \otimes R)&=(R \otimes I_d^{\otimes l-m} \otimes R^2)\\
    (R^{-1}\otimes I_d^{\otimes l})(R^2 \otimes I_d^{\otimes l-m} \otimes R)&=(R^{-1}\otimes I_d^{\otimes l})(R \otimes I_d^{\otimes l-m} \otimes R^2)\\
    (R \otimes I_d^{\otimes l-m} \otimes R)&=(I_d^{\otimes m} \otimes I_d^{\otimes l-m} \otimes R^2)\\
    (R \otimes I_d^{\otimes l-m} \otimes R)(I_d^{\otimes l}\otimes R^{-1})&=(I_d^{\otimes m} \otimes I_d^{\otimes l-m} \otimes R^2)(I_d^{\otimes l}\otimes R^{-1})\\
    (R \otimes I_d^{\otimes l-m} \otimes I_d^{\otimes m})&=(I_d^{\otimes m} \otimes I_d^{\otimes l-m} \otimes R)\\
    (R \otimes I_d^{\otimes l} )&=(I_d^{\otimes l} \otimes R)\\
\end{align*}
It follows that $R$ must be a scalar multiple of the identity $I_{d^m}$ in order to solve the $(d,m,l)$-gYBE whenever $l \geq m$.
\end{proof}


\subsection{The Algebraic Yang-Baxter Equation}
\label{secaYBE}
 
\begin{definition}

Let $P$ be the swap operator interchanging qudits denoted in the Dirac notation \cite{dirac_new_1939} by $P: \ket{ij} \rightarrow \ket{ji}$ and let $I$ be the $d\times d$ identity matrix. When $R$ is a $d^2 \times d^2$ define

\begin{align}
    R_{12} &= (R\otimes I)\\ R_{13} &= (I\otimes P)(R\otimes I)(I\otimes P)\\
    R_{23} &= (I\otimes R)
\end{align}

the \textit{algebraic Yang-Baxter equation} (aYBE) is then defined by:

\begin{equation}
    \label{AlgebraicYBE}
    R_{12}R_{13}R_{23} = R_{23}R_{13}R_{12}
\end{equation}

\end{definition}

The aYBE and bYBE are closely related: if $R$ solves the aYBE then $RP$ solves the bYBE and vice versa. The matrix form of the aYBE can be written as a system of polynomial equations (see \cite{hietarinta_solving_1993} or \cite{dye_unitary_2003} or the appendix):
\begin{align}
    \label{eq:aYBE}
    R_{j_1j_2}^{k_1k_2}R_{k_1j_3}^{l_1k_3}R_{k_2k_3}^{l_2l_3} = R_{j_2j_3}^{k_2k_3}R_{j_1k_3}^{k_1l_3}R_{k_1k_2}^{l_1l_2}
\end{align}where each equation is indexed by $(j_1,j_2,j_3,l_1,l_2,l_3)$, with each index ranging from $1$ to $d$, and following the Einstein summation convention used in differential geometry \cite{einstein_grundlage_1916}, sums are taken over repeated indicies. In this case, the indicies denoted with a $k$ ($k_p$ for $p=1,2,3$) are summed over. This indexing appears with some variation between references because the aYBE is invariant under various index changes, as listed in \cite{hietarinta_solving_1993}. We will use the following convention:
\begin{align}
R(e_i\otimes e_i) &= \sum R_{ij}^{ab} e_a \otimes e_b 
\end{align}
For example when $d=3$ the matrix $R$ looks like:

\begin{align}
\begin{pmatrix}
R_{11}^{11} & R_{12}^{11} & R_{13}^{11} & R_{21}^{11} & R_{22}^{11} & R_{23}^{11} & R_{31}^{11} &
   R_{32}^{11} & R_{33}^{11} \\
 R_{11}^{12} & R_{12}^{12} & R_{13}^{12} & R_{21}^{12} & R_{22}^{12} & R_{23}^{12} & R_{31}^{12} &
   R_{32}^{12} & R_{33}^{12} \\
 R_{11}^{13} & R_{12}^{13} & R_{13}^{13} & R_{21}^{13} & R_{22}^{13} & R_{23}^{13} & R_{31}^{13} &
   R_{32}^{13} & R_{33}^{13} \\
 R_{11}^{21} & R_{12}^{21} & R_{13}^{21} & R_{21}^{21} & R_{22}^{21} & R_{23}^{21} & R_{31}^{21} &
   R_{32}^{21} & R_{33}^{21} \\
 R_{11}^{22} & R_{12}^{22} & R_{13}^{22} & R_{21}^{22} & R_{22}^{22} & R_{23}^{22} & R_{31}^{22} &
   R_{32}^{22} & R_{33}^{22} \\
 R_{11}^{23} & R_{12}^{23} & R_{13}^{23} & R_{21}^{23} & R_{22}^{23} & R_{23}^{23} & R_{31}^{23} &
   R_{32}^{23} & R_{33}^{23} \\
 R_{11}^{31} & R_{12}^{31} & R_{13}^{31} & R_{21}^{31} & R_{22}^{31} & R_{23}^{31} & R_{31}^{31} &
   R_{32}^{31} & R_{33}^{31} \\
 R_{11}^{32} & R_{12}^{32} & R_{13}^{32} & R_{21}^{32} & R_{22}^{32} & R_{23}^{32} & R_{31}^{32} &
   R_{32}^{32} & R_{33}^{32} \\
 R_{11}^{33} & R_{12}^{33} & R_{13}^{33} & R_{21}^{33} & R_{22}^{33} & R_{23}^{33} & R_{31}^{33} &
   R_{32}^{33} & R_{33}^{33} \\
\end{pmatrix}
\end{align}

\subsection{Symmetries of the Yang-Baxter Equations}

Each solution of a gYBE generates more solutions under the following symmetries, to the same gYBE (and with the appropriate choice of $(d,m,l)$ these symmetries also generate more solutions to the bYBE and aYBE).

\begin{proposition}
\label{symmetries}
If $R$ is a solution to the $(d,m,l)$-gYBE then the following are also solutions:
\begin{enumerate}
\item $\lambda R$ for any nonzero scalar $\lambda$
\item $R^{-1}$
\item The complex conjugate $R^*$
\item The transpose $R^T$ and hence the complex adjoint $R^\dagger$
\item $ Q^{\otimes m}R (Q^{-1})^{\otimes m} $ where $Q$ is any complex non-singular $d\times d$ matrix.
\end{enumerate}
\end{proposition}

These symmetries are well known in the case of both the bYBE and aYBE and appear in \cite{kassel_quantum_1995} and \cite{hietarinta_solving_1993}, a proof that these also apply to the gYBEs is provided in the appendix.

\section{Universal Quantum Gates}
Recall as in \cite{brylinski_universal_2002} that a unitary matrix (acting on $n$ qudits) is \textit{universal} for quantum computation if it together with all \textit{local} unitary transformations from $V\rightarrow V$ (single qudit gates) generate a dense subgroup of $U(d^n)$. A unitary matrix is \textit{exactly universal} if the full group $U(d^n)$ is generated. Brylinski showed in \cite{brylinski_universal_2002} that a two qudit gate $U$ is universal if and only if it is entangling. That is, if there is a state $\ket{ij} \in V\otimes V$ such that $U \ket{ij}$ cannot be written as the tensor product of two qubits. For example the CNOT gate defined by $\ket{ij}\rightarrow \ket{i,i\oplus j}$ is exactly universal since it sends the state $\ket{00}+\ket{10} $ to $ \ket{00}+\ket{11}$ \cite{brylinski_universal_2002}. 

\begin{equation}
\begin{pmatrix}
1 & 0 & 0 & 0\\
0 & 1 & 0 & 0\\
0 & 0 & 0 & 1\\
0 & 0 & 1 & 0
\end{pmatrix}
\begin{pmatrix}
1\\
0\\
1\\
0\\
\end{pmatrix}
=
\begin{pmatrix}
1\\
0\\
0\\
1\\
\end{pmatrix}
\end{equation}

The entangling condition also provides a connection between topological entanglement and quantum entanglement. Any $R$ matrix gives rise to a knot and link invariant \cite{turaev_yang-baxter_1988}, and if $R$ is not entangling it can not be used to distinguish between two knots \cite{alagic_yangbaxter_2016}. While almost every unitary gate is universal \cite{deutsch_universality_1995}, finding matrices that are both universal and solve the Yang-Baxter equation is a complex task. It is shown by Kauffman and Lomonaco in \cite{kauffman_braiding_2004} that the following unitary solutions to the Yang-Baxter equation are also exactly universal as two qubit gates.

\begin{gather*}
\begin{pmatrix}
1 & 0 & 0 & 0\\
0 & 1 & 0 & 0\\
0 & 0 & 1 & 0\\
0 & 0 & 0 & -1
\end{pmatrix}
\qquad 
\begin{pmatrix}
1 & 0 & 0 & 0\\
0 & 0 & 1 & 0\\
0 & 1 & 0 & 0\\
0 & 0 & 0 & -1
\end{pmatrix}
\qquad 
\frac{1}{\sqrt{2}}
\begin{pmatrix}
1 & 0 & 0 & 1\\
0 & 1 & 1 & 0\\
0 & 1 & -1 & 0\\
-1 & 0 & 0 & 1
\end{pmatrix}
\end{gather*}

The CNOT gate generalizes to the controlled increment gate $C^n_{X,d}$, which is defined recursively in \cite{hunt_grovers_2020}. This is recalled next. Let $n$ be the number of qudits being acted on. Let $X_d$ denote the $d \times d$ increment gate (INC):

\begin{align}
    X_d &= \begin{pmatrix} 
    0 & 0 & \dots  & 0 & 1\\
    1 & 0 & \dots  & 0 & 0\\
    0 & 1 & \dots  & 0 & 0\\
    \vdots & \vdots & \ddots & \vdots & \vdots \\
   0 & 0 & \dots  & 1 & 0
    \end{pmatrix}
\end{align}

CNOT (or sometimes CINC, the controlled increment gate) is then defined as in \cite{hunt_grovers_2020} by letting $C^1_{X,d} = X_d$ and then recursively constructing

\begin{align}
    C^n_{X,d} &= \begin{pmatrix} 
    I & 0 & 0 & \dots  & 0\\
    0 & C^{n-1}_{X,d} & 0 & \dots  & 0\\
    0 & 0 & I & \dots  & 0\\
    \vdots & \vdots & \vdots & \ddots & \vdots \\
   0 & 0 & 0 & \dots  & I
    \end{pmatrix}
\end{align}

Note that the identity matrix block is repeated $n-2$ times in the lower right. Induction on $n$ shows that $C^n_{X,d}$ is a real unitary matrix for all $n$.

\begin{proof}

When $n=1$ we have $C^1_{X,d} = X_d$ and $X_d X_d^T = I$. Now suppose $C^n_{X,d}$ is unitary up to some $n>1$.
    
\begin{align*}
    C^n_{X,d} (C^n_{X,d})^\dagger &= \begin{pmatrix} 
    I & 0 & 0 & \dots  & 0\\
    0 & C^{n-1}_{X,d} & 0 & \dots  & 0\\
    0 & 0 & I & \dots  & 0\\
    \vdots & \vdots & \vdots & \ddots & \vdots \\
   0 & 0 & 0 & \dots  & I
    \end{pmatrix}
    \begin{pmatrix} 
    I & 0 & 0 & \dots  & 0\\
    0 & (C^{n-1}_{X,d})^\dagger & 0 & \dots  & 0\\
    0 & 0 & I & \dots  & 0\\
    \vdots & \vdots & \vdots & \ddots & \vdots \\
   0 & 0 & 0 & \dots  & I
    \end{pmatrix}\\
    &=\begin{pmatrix} 
    I & 0 & 0 & \dots  & 0\\
    0 & C^{n-1}_{X,d}(C^{n-1}_{X,d})^\dagger & 0 & \dots  & 0\\
    0 & 0 & I & \dots  & 0\\
    \vdots & \vdots & \vdots & \ddots & \vdots \\
   0 & 0 & 0 & \dots  & I
    \end{pmatrix}=\begin{pmatrix}
    I & 0 & 0 & \dots  & 0\\
    0 & I & 0 & \dots  & 0\\
    0 & 0 & I & \dots  & 0\\
    \vdots & \vdots & \vdots & \ddots & \vdots \\
   0 & 0 & 0 & \dots  & I
    \end{pmatrix}
\end{align*}
\end{proof}

The CNOT gate appears in this family as $C^2_{X,d}$. Since the CNOT gate $C_{X,d}^2$ along with all 1-qudit gates is an exactly universal gate set, another way to prove a particular matrix is exactly universal is to show that CNOT can be expressed using that matrix and the Kronecker product of local unitary matrices \cite{brylinski_universal_2002}. This method also has the advantage of immediately providing a way to re-express any algorithm that is written in terms of CNOT and local gates.

\subsection{A Universal Yang-Baxter Operator}

Unitary representations of the braid group serve as the gates in a topological quantum computer \cite{kauffman_braiding_2004}. It is therefore desirable to find unitary solutions the bYBE or gYBE which are also universal for quantum computation. Finding solutions to any of the variation Yang-Baxter equation is generally a difficult task which reduces to solving a large system of multivariate polynomial equations. In particular, the bYBE, when the vector space $V$ has dimension $d$, involves solving a system of $d^6$ cubic equations in $d^4$ variables. What follows is one of the primary contributions of this paper: a not previously noted universal unitary solution to the bYBE in all dimensions. To construct the solution, first recall the discrete quantum Fourier transform denoted $F_d$ and with $\omega = e^{i2\pi/d}$:
\begin{align}
    F_d &= \frac{1}{\sqrt{d}}\begin{bmatrix} 
    1 & 1 & 1 & \dots  & 1\\
    1 & \omega & \omega^2 & \dots  & \omega^{d-1}\\
    1 & \omega^2 & \omega^4 & \dots  & \omega^{2(d-1)}\\
    \vdots & \vdots & \vdots & \ddots & \vdots \\
    1 & \omega^{d-1} & \omega^{2(d-1)} & \dots  & \omega^{(d-1)(d-1)}
    \end{bmatrix}
\end{align}

To construct the bYBE solution we define the following unitary solution to the aYBE:

\begin{align}
\label{diagSol}
    R_d&=(I\otimes F_d) C^2_{X,d}  (I\otimes F_d^{\dagger})
\end{align}
In theorem \ref{diagonalTheorem} we show that this solution is part of a larger family of diagonal unitary solutions to the aYBE, any of which can be converted to a solution of the bYBE by composing with the swap operator $P$. 
\begin{theorem}
\label{diagonalTheorem}
All $d^2 \times d^2$ diagonal matrices are solutions to the aYBE in dimension $d$ and in particular $R_d$ provides an example of an exactly universal unitary solution to the aYBE.
\end{theorem}

\begin{proof}
First note that $R_d$ is unitary since it is the product of unitary matrices. To express $C^2_{X,d}$ in terms of $R_d$ and the Kronecker product of local linear transformations we can conjugate $R_d$ by $(I\otimes F_d)$ as follows:

\begin{align*}
    C^2_{X,d} &= (I \otimes F_d^\dagger) R_d (I \otimes F_d)
\end{align*} 

To see that $R_d$ is a solution to the algebraic Yang-Baxter equation we can simplify the form of $R_d$ using block matrix multiplication:

\begin{align*}
    R_d&=(I\otimes F_d) C^2_{X,d}  (I\otimes F_d^{\dagger})\\
    &= \begin{pmatrix} 
    F_d & 0 & 0 & \dots  & 0\\
    0 & F_d & 0 & \dots  & 0\\
    0 & 0 & F_d & \dots  & 0\\
    \vdots & \vdots & \vdots & \ddots & \vdots \\
   0 & 0 & 0 & \dots  & F_d
    \end{pmatrix}
    \begin{pmatrix} 
    I & 0 & 0 & \dots  & 0\\
    0 & X_d & 0 & \dots  & 0\\
    0 & 0 & I & \dots  & 0\\
    \vdots & \vdots & \vdots & \ddots & \vdots \\
   0 & 0 & 0 & \dots  & I
    \end{pmatrix}
    \begin{pmatrix} 
    F_d^\dagger & 0 & 0 & \dots  & 0\\
    0 & F_d^\dagger & 0 & \dots  & 0\\
    0 & 0 & F_d^\dagger & \dots  & 0\\
    \vdots & \vdots & \vdots & \ddots & \vdots \\
   0 & 0 & 0 & \dots  & F_d^\dagger
    \end{pmatrix}\\
    &=\begin{pmatrix} 
    I & 0 & 0 & \dots  & 0\\
    0 & F_d X_d F_d^\dagger& 0 & \dots  & 0\\
    0 & 0 & I & \dots  & 0\\
    \vdots & \vdots & \vdots & \ddots & \vdots \\
   0 & 0 & 0 & \dots  & I
    \end{pmatrix}
\end{align*}

This matrix turns out to be diagonal since the Fourier transform diagonalizes $X_d$. The characteristic polynomial of $X_d$ is $\lambda^d-1$ and therefore the eigenvalues of $X_d$ are the $d$ roots of unity: $e^{\frac{k 2 \pi i}{d}}$ for $k=1\dots d$. It is then straightforward to check that the corresponding eigenvectors are given by the columns of $F_d$.

 To see that this solves the aYBE, we prove the stronger result that any diagonal matrix is a solution to the aYBE. Using the indexing convention from described in section \ref{secaYBE}, an invertible diagonal matrix will have nonzero entries along the diagonal: $R_{ab}^{ab}\neq 0$ and zero entries elsewhere. Now consider one of the equations indexed by $(j_1,j_2,j_3,l_1,l_2,l_3)$. All terms in the sum on the left side of equation \ref{eq:aYBE} will vanish unless each variable is from the diagonal of $R$, that is unless $k_1=j_1$, $l_1=k_1$, $l_2=k_2$, $k_2=j_2$, $k_3=j_3$, $l_3=k_3$. By the same reasoning the terms in the sum on the right hand side will vanish unless $k_2=j_2$, $k_1=j_1$, $l_1=k_1$, $l_2=k_2$, $j_3=k_3$, $k_3=l_3$, $k_2=l_2$. By considering only the nonzero terms of the sum, the algebraic Yang-Baxter equation is satisfied whenever $R$ is a diagonal matrix:
\begin{align*}
    R_{j_1l_2}^{j_1l_2}R_{j_1j_3}^{j_1j_3}R_{j_2j_3}^{j_2j_3} = R_{j_2j_3}^{j_2j_3}R_{j_1j_3}^{j_1j_3}R_{j_1l_2}^{j_1l_2}
\end{align*}
\end{proof}
As a consequence $R_d P$ is a unitary solution to the bYBE, and by the symmetries above $(Q\otimes Q)R_d(Q\otimes Q)^{-1}$ is also a universal unitary solution to the Yang Baxter equation whenever $Q$ is a complex $d\times d$ unitary matrix. This provides a way to generate many non-trivial examples of unitary solutions to the Yang-Baxter equation, which are also universal as quantum gates, and provide an explicit decomposition of the CNOT gate. For example, when $d=3$, if $Q=F_3^3$ we obtain the following unitary solution to the bYBE:

\begin{equation*}
\resizebox{\textwidth}{!}{%
$
\frac{1}{6}\begin{pmatrix}
\begin{array}{ccccccccc}
 4 & 1-i \sqrt{3} & 1+i \sqrt{3} & 0 & 0 & 0 & 2 & -1+i \sqrt{3} & -1-i \sqrt{3} \\
 2 & -1+i \sqrt{3} & -1-i \sqrt{3} & 4 & 1-i \sqrt{3} & 1+i \sqrt{3} & 0 & 0 & 0 \\
 0 & 0 & 0 & 2 & -1+i \sqrt{3} & -1-i \sqrt{3} & 4 & 1-i \sqrt{3} & 1+i \sqrt{3} \\
 1+i \sqrt{3} & 4 & 1-i \sqrt{3} & 0 & 0 & 0 & -1-i \sqrt{3} & 2 & -1+i \sqrt{3} \\
 -1-i \sqrt{3} & 2 & -1+i \sqrt{3} & 1+i \sqrt{3} & 4 & 1-i \sqrt{3} & 0 & 0 & 0 \\
 0 & 0 & 0 & -1-i \sqrt{3} & 2 & -1+i \sqrt{3} & 1+i \sqrt{3} & 4 & 1-i \sqrt{3} \\
 1-i \sqrt{3} & 1+i \sqrt{3} & 4 & 0 & 0 & 0 & -1+i \sqrt{3} & -1-i \sqrt{3} & 2 \\
 -1+i \sqrt{3} & -1-i \sqrt{3} & 2 & 1-i \sqrt{3} & 1+i \sqrt{3} & 4 & 0 & 0 & 0 \\
 0 & 0 & 0 & -1+i \sqrt{3} & -1-i \sqrt{3} & 2 & 1-i \sqrt{3} & 1+i \sqrt{3} & 4 \\
\end{array}
\end{pmatrix}$%
}
\end{equation*}

As any diagonal matrix solves the aYBE, other universal unitary solutions to the bYBE can be generated in a similar way. The conditions under which an arbitrary diagonal matrix is universal is provided in \cite{brylinski_universal_2002}. While we've seen that any diagonal matrix is a solution to the aYBE, the bYBE only has one family of diagonal solutions.

\begin{theorem}
The only invertible diagonal solutions to the braided Yang-Baxter equation are scalar multiples of the identity matrix.
\end{theorem}
\begin{proof}
Let $R = diag(r_1, \dots, r_{d^2})$ be an invertible diagonal matrix. Since both $(R\otimes I)$ and $(I\otimes R)$ are diagonal and will commute, the left hand side of the braided Yang-Baxter equation becomes $(R \otimes I)(I \otimes R) (R\otimes I) = (R^2\otimes I)(I \otimes R)$ and the right hand side can be written $(I \otimes R) (R\otimes I) (I \otimes R) = (R\otimes I)(I\otimes R^2)$. The bYBE then becomes: 
\begin{align*}
    (I\otimes R^2)(R\otimes I)&=(I \otimes R)(R^2\otimes I)\\
    (I\otimes R)&=(R\otimes I)\\
\end{align*}
Examining the terms on the diagonal of $(I\otimes R)-(R\otimes I)$ we get the equations: $r_2-r_1=0$, $r_3-r_1 =0$, $\dots$ , $r_{d^2}-r_1 =0$.
Therefore the only invertible solution is when $r_1=r_2=\dots=r_{d^2}$
\end{proof}

\section{$X$-shaped solutions to the gYBE}

The permutation solutions to a gYBE can be found by brute force computation in low dimensions. From a permutation matrix solution one can construct a \textit{monomial} solution by replacing the 1's with variables and solving for the conditions under which the new matrix is a solution. The $(2,3,1)$ and $(2,3,2)$ monomial solutions have been classified fully in \cite{nemec_monomial_2014}. In contrast to the bYBE, the $(2,3,1)$ and $(2,3,2)$ gYBE's don't have any monomial solutions with $d$ free parameters. Other than the monomial solutions in \cite{nemec_monomial_2014}, there are currently only a few known solutions to the non-bYBE $(d,m,l)$-gYBE up to the symmetries in proposition \ref{symmetries}. One of the well known solutions is the \textit{X-shape} solution to the $(2,3,2)$-gYBE that appears in \cite{rowell_extraspecial_2010}:

\begin{align*}
R_X=
\frac{1}{\sqrt{2}}
\begin{pmatrix}
1 & 0 & 0 & 0 & 0 & 0 & 0 & 1\\
0 & 1 & 0 & 0 & 0 & 0 & 1 & 0\\
0 & 0 & 1 & 0 & 0 & 1 & 0 & 0\\
0 & 0 & 0 & 1 & 1 & 0 & 0 & 0\\
0 & 0 & 0 & -1& 1 & 0 & 0 & 0\\
0 & 0 & -1 & 0 & 0 & 1 & 0 & 0\\
0 & -1 & 0 & 0 & 0 & 0 & 1 & 0\\
-1 & 0 & 0 & 0 & 0 & 0 & 0 & 1\\
\end{pmatrix}
\end{align*}

After that solution was found a handful of other solutions were found in \cite{chen_generalized_2012}, \cite{rowell_parameter-dependent_2014}, \cite{kitaev_solutions_2012}. The question of finding all unitary solutions with nonzero entries in the same position as the nonzero entries of $R_X$ is posed in \cite{chen_generalized_2012}. This question demonstrates another more general method used to find solutions: pick an ansatz or initial guess about the form of the matrix, and in some cases, this will simplify the system of polynomial equations enough that they can be fully solved. An X-shaped solution will be universal since it entangles the state $\ket{0}$.

\subsection{(2,3,2)-gYBE X-Shaped Solutions}
What follows is a description of the procedure used to obtain all 4 distinct families of X-shaped solutions to the $(2,3,2)$-gYBE. 

\begin{align}
X=\begin{pmatrix}
r_{11} & 0 & 0 & 0 & 0 & 0 & 0 & r_{18} \\
 0 & r_{22} & 0 & 0 & 0 & 0 & r_{27} & 0 \\
 0 & 0 & r_{33} & 0 & 0 & r_{36} & 0 & 0 \\
 0 & 0 & 0 & r_{44} & r_{45} & 0 & 0 & 0 \\
 0 & 0 & 0 & r_{54} & r_{55} & 0 & 0 & 0 \\
 0 & 0 & r_{63} & 0 & 0 & r_{66} & 0 & 0 \\
 0 & r_{72} & 0 & 0 & 0 & 0 & r_{77} & 0 \\
 r_{81} & 0 & 0 & 0 & 0 & 0 & 0 & r_{88} \\
\end{pmatrix}
\end{align}

The general $(2,3,2)$ X-shaped ansatz above generates a set of $116$ polynomial equations in $16$ variables determined by the gYBE. The variable $r_{22}$ appears in the most equations and can be scaled to 1 using the overall scaling symmetry since all variables are assumed to be nonzero. After this scaling, the following equations are in the set:

\begin{align*}
    -r_{36}r_{63} (r_{55}-1) &=0\\
    r_{36} r_{63} (r_{44} - r_{77})&=0
\end{align*}
and therefore $r_{55}=1$ and $r_{77}=r_{44}$. After making these substitutions $108$ equations remain. This system is then small enough that a Gr\"{o}bner basis can be computed using a computer algebra system. We used Mathematica to compute a Gr\"{o}bner basis with a lexicographic monomial ordering, and then used the reduce function to find all solutions to the system. Initially this procedure results in $7$ solutions, which are listed below, with $ \alpha  ,\beta , \gamma, \delta$ being complex free parameters:

\setlength\arraycolsep{1.5pt}

\noindent \begin{minipage}{.5\linewidth}
\begin{align*}
X_1=& \begin{pmatrix}
 \delta & 0 & 0 & 0 & 0 & 0 & 0 & \frac{ \alpha \beta^2}{\delta^2} \\
 0 & 1 & 0 & 0 & 0 & 0 &  \beta & 0 \\
 0 & 0 & \delta & 0 & 0 &  \alpha  & 0 & 0 \\
 0 & 0 & 0 & 1 & \frac{\delta^2}{ \beta} & 0 & 0 & 0 \\
 0 & 0 & 0 &  \beta & 1 & 0 & 0 & 0 \\
 0 & 0 & \frac{1}{\alpha} & 0 & 0 & \delta & 0 & 0 \\
 0 & \frac{\delta^2}{ \beta} & 0 & 0 & 0 & 0 & 1 & 0 \\
 \frac{\delta^2}{\alpha  \beta^2} & 0 & 0 & 0 & 0 & 0 & 0 & \delta\\
\end{pmatrix}\\
X_2=& \begin{pmatrix}
 -i & 0 & 0 & 0 & 0 & 0 & 0 & i \alpha  \beta^2 \\
 0 & 1 & 0 & 0 & 0 & 0 &  \beta & 0 \\
 0 & 0 & 1 & 0 & 0 & \alpha & 0 & 0 \\
 0 & 0 & 0 & -i & \frac{i}{ \beta} & 0 & 0 & 0 \\
 0 & 0 & 0 &  \beta & 1 & 0 & 0 & 0 \\
 0 & 0 & \frac{i}{\alpha} & 0 & 0 & -i & 0 & 0 \\
 0 & \frac{i}{ \beta} & 0 & 0 & 0 & 0 & -i & 0 \\
 \frac{1}{\alpha  \beta^2} & 0 & 0 & 0 & 0 & 0 & 0 & 1 \\
\end{pmatrix} \\
X_3=& \begin{pmatrix}
 1 & 0 & 0 & 0 & 0 & 0 & 0 & \alpha  \beta^2 \\
 0 & 1 & 0 & 0 & 0 & 0 & - \beta & 0 \\
 0 & 0 & 1 & 0 & 0 & \alpha & 0 & 0 \\
 0 & 0 & 0 & 1 & -\frac{1}{ \beta} & 0 & 0 & 0 \\
 0 & 0 & 0 &  \beta & 1 & 0 & 0 & 0 \\
 0 & 0 & -\frac{1}{\alpha} & 0 & 0 & 1 & 0 & 0 \\
 0 & \frac{1}{ \beta} & 0 & 0 & 0 & 0 & 1 & 0 \\
 -\frac{1}{\alpha  \beta^2} & 0 & 0 & 0 & 0 & 0 & 0 & 1 \\
\end{pmatrix}
\end{align*}
\end{minipage}%
\begin{minipage}{.5\linewidth}
\begin{align*}
X_5=& \begin{pmatrix}
 i & 0 & 0 & 0 & 0 & 0 & 0 & -i \alpha  \beta^2 \\
 0 & 1 & 0 & 0 & 0 & 0 &  \beta & 0 \\
 0 & 0 & 1 & 0 & 0 & \alpha & 0 & 0 \\
 0 & 0 & 0 & i & -\frac{i}{ \beta} & 0 & 0 & 0 \\
 0 & 0 & 0 &  \beta & 1 & 0 & 0 & 0 \\
 0 & 0 & -\frac{i}{\alpha} & 0 & 0 & i & 0 & 0 \\
 0 & -\frac{i}{ \beta} & 0 & 0 & 0 & 0 & i & 0 \\
 \frac{1}{\alpha  \beta^2} & 0 & 0 & 0 & 0 & 0 & 0 & 1 \\
\end{pmatrix} \\
X_6=& \begin{pmatrix}
1 & 0 & 0 & 0 & 0 & 0 & 0 & \alpha  \beta^2 \\
 0 & 1 & 0 & 0 & 0 & 0 &  \beta & 0 \\
 0 & 0 & 1 & 0 & 0 & \alpha & 0 & 0 \\
 0 & 0 & 0 & 1 & \frac{1}{ \beta} & 0 & 0 & 0 \\
 0 & 0 & 0 &  \beta & 1 & 0 & 0 & 0 \\
 0 & 0 & \frac{1}{\alpha} & 0 & 0 & 1 & 0 & 0 \\
 0 & \frac{1}{ \beta} & 0 & 0 & 0 & 0 & 1 & 0 \\
 \frac{1}{a  \beta^2} & 0 & 0 & 0 & 0 & 0 & 0 & 1 \\
\end{pmatrix}  \\
X_7=& \begin{pmatrix}
 1 & 0 & 0 & 0 & 0 & 0 & 0 & \frac{\alpha  \beta^2}{\gamma} \\
 0 & 1 & 0 & 0 & 0 & 0 &  \beta & 0 \\
 0 & 0 & \gamma & 0 & 0 & \alpha & 0 & 0 \\
 0 & 0 & 0 & \gamma & \frac{\gamma}{ \beta} & 0 & 0 & 0 \\
 0 & 0 & 0 &  \beta & 1 & 0 & 0 & 0 \\
 0 & 0 & \frac{\gamma}{\alpha} & 0 & 0 & 1 & 0 & 0 \\
 0 & \frac{\gamma}{ \beta} & 0 & 0 & 0 & 0 & \gamma & 0 \\
 \frac{\gamma^2}{\alpha  \beta^2} & 0 & 0 & 0 & 0 & 0 & 0 & \gamma \\
\end{pmatrix}
\end{align*}
\end{minipage}

\begin{align*}
X_4=& \begin{pmatrix}
 2-\delta & 0 & 0 & 0 & 0 & 0 & 0 & \frac{\alpha \beta^2}{\delta^2-2 \delta+2} \\
 0 & 1 & 0 & 0 & 0 & 0 &  \beta & 0 \\
 0 & 0 & 2-\delta & 0 & 0 & \alpha & 0 & 0 \\
 0 & 0 & 0 & 1 & \frac{\delta^2-2 \delta+2}{ \beta} & 0 & 0 & 0 \\
 0 & 0 & 0 &  \beta & 1 & 0 & 0 & 0 \\
 0 & 0 & \frac{1}{\alpha} & 0 & 0 & \delta & 0 & 0 \\
 0 & \frac{\delta^2-2 \delta+2}{ \beta} & 0 & 0 & 0 & 0 & 1 & 0 \\
 \frac{\delta^2-2 \delta+2}{\alpha \beta^2} & 0 & 0 & 0 & 0 & 0 & 0 & \delta \\
\end{pmatrix}
\end{align*}

The matrices $X_6$ and $X_7$ are not invertible for any choice of the parameters. The matrices $X_2$ and $X_5$ are from the same family after utilizing the symmetries in proposition \ref{symmetries}. In particular, $X_5=2*X_2^{-1}$ with $ \alpha $ replaced with $i \alpha$ and $ \beta$ replaced with $i \beta$. We now determine if there is a choice of the parameters and overall scale factor $\lambda$ which make $X_1$, $X_2$, $X_3$, $X_4$ unitary.

\begin{proposition}
$\lambda X_1$ is unitary when $\operatorname{Re}(\delta)=0$, $\delta^2 = -| \beta|^2$, $|  \alpha  |=1$, $|\lambda|^2 = \frac{1}{1+|\delta|^2}$
\end{proposition}
\begin{proof}
\begin{align*}
\lambda X_1 (\lambda X_1)^{\dagger}=& \lambda \bar{\lambda} \begin{pmatrix}
\frac{a  \beta^2 \bar{  \alpha  } \bar{ \beta}^2}{\delta^2 \bar{\delta}^2}+\delta \bar{\delta} & 0 & 0 & 0 & 0 & 0 & 0 &
   \frac{  \alpha    \beta^2 \bar{\delta}}{\delta^2}+\frac{\delta \bar{\delta}^2}{\bar{  \alpha  } \bar{ \beta}^2} \\
 0 &  \beta \bar{ \beta}+1 & 0 & 0 & 0 & 0 & \frac{\bar{\delta}^2}{\bar{ \beta}}+ \beta & 0 \\
 0 & 0 &   \alpha   \bar{  \alpha  }+\delta \bar{\delta} & 0 & 0 & \frac{\delta}{\bar{  \alpha  }}+  \alpha   \bar{\delta} & 0 & 0 \\
 0 & 0 & 0 & \frac{\delta^2 \bar{\delta}^2}{ \beta \bar{ \beta}}+1 & \bar{ \beta}+\frac{\delta^2}{ \beta} & 0 & 0 & 0 \\
 0 & 0 & 0 & \frac{\bar{\delta}^2}{\bar{ \beta}}+ \beta &  \beta \bar{ \beta}+1 & 0 & 0 & 0 \\
 0 & 0 & \delta \bar{  \alpha  }+\frac{\bar{\delta}}{a} & 0 & 0 & \frac{1}{  \alpha   \bar{  \alpha  }}+\delta \bar{\delta} & 0 & 0 \\
 0 & \bar{ \beta}+\frac{\delta^2}{ \beta} & 0 & 0 & 0 & 0 & \frac{\delta^2 \bar{\delta}^2}{ \beta \bar{ \beta}}+1 & 0 \\
 \frac{\delta^2 \bar{\delta}}{  \alpha    \beta^2}+\frac{\delta \bar{  \alpha  } \bar{ \beta}^2}{\bar{\delta}^2} & 0 & 0 & 0 & 0 & 0 & 0
   & \frac{\delta^2 \bar{\delta}^2}{  \alpha    \beta^2 \bar{  \alpha  } \bar{ \beta}^2}+\delta \bar{\delta} \\
\end{pmatrix}
\end{align*}

Since it is required that $\bar{ \beta}+\frac{\delta^2}{ \beta}=0$ and $\frac{\delta}{\bar{ \alpha }} +  \alpha \bar{\delta}=0$ we must have $\delta^2=- \beta\bar{ \beta} = -| \beta|^2$ and $\alpha \bar{ \alpha }=\frac{-\delta}{\bar{\delta}}$. Therefore $\operatorname{Re}(\delta)=0$ and $\alpha \bar{ \alpha }=1$ and all off-diagonal elements will be zero when this is the case:
\begin{align*}
\frac{\delta^2 \bar{\delta}}{ \alpha   \beta^2}+\frac{\delta \bar{ \alpha } \bar{ \beta}^2}{\bar{\delta}^2} &=
\delta^2 \bar{\delta}+\frac{\delta  \alpha  \bar{ \alpha } ( \beta\bar{ \beta})^2}{\bar{\delta}^2} =
\delta^2 \bar{\delta}+\frac{\delta  \alpha  \bar{ \alpha } (-\delta^2)^2}{\bar{\delta}^2}\\ &=
\delta^2 \bar{\delta}^3+  \alpha  \bar{ \alpha } \delta^5 =
\bar{\delta}^3+  \alpha  \bar{ \alpha } \delta^3 =
\bar{\delta}^3+ \frac{-\delta}{\bar{\delta}} \delta^3 =
(-\delta)^4 -\delta^4 =0
\end{align*}
A similar computation shows that all diagonal elements are equal to $1+|\delta|^2$ and therefore $\lambda X_1$ will be unitary whenever $|\lambda|^2 = \frac{1}{1+|\delta|^2}$.
\end{proof}

\begin{proposition}
$\lambda X_2$ is unitary when $| \alpha |=1$, $| \beta|=1$, and $|\lambda|^2 = \frac{1}{2}$
\end{proposition}
\begin{proof}
\begin{align*}
\lambda X_2 (\lambda X_2)^{\dagger}=& \lambda \bar{\lambda} \begin{pmatrix}
  \alpha   \beta^2 \bar{ \alpha } \bar{ \beta}^2+1 & 0 & 0 & 0 & 0 & 0 & 0 & i  \alpha   \beta^2-\frac{i}{\bar{ \alpha } \bar{ \beta}^2}
   \\
 0 &  \beta \bar{ \beta}+1 & 0 & 0 & 0 & 0 & i  \beta-\frac{i}{\bar{ \beta}} & 0 \\
 0 & 0 &  \alpha  \bar{ \alpha }+1 & 0 & 0 & i \alpha-\frac{i}{\bar{ \alpha }} & 0 & 0 \\
 0 & 0 & 0 & \frac{1}{ \beta \bar{ \beta}}+1 & \frac{i}{ \beta}-i \bar{ \beta} & 0 & 0 & 0 \\
 0 & 0 & 0 & i  \beta-\frac{i}{\bar{ \beta}} &  \beta \bar{ \beta}+1 & 0 & 0 & 0 \\
 0 & 0 & \frac{i}{ \alpha }-i \bar{ \alpha } & 0 & 0 & \frac{1}{\alpha \bar{\alpha}}+1 & 0 & 0 \\
 0 & \frac{i}{ \beta}-i \bar{ \beta} & 0 & 0 & 0 & 0 & \frac{1}{ \beta \bar{ \beta}}+1 & 0 \\
 \frac{i}{ \alpha   \beta^2}-i \bar{\alpha} \bar{ \beta}^2 & 0 & 0 & 0 & 0 & 0 & 0 & \frac{1}{ \alpha   \beta^2 \bar{ \alpha }
   \bar{ \beta}^2}+1 \\
\end{pmatrix}
\end{align*}
To make the off-diagonal elements zero we need $\frac{i}{ \beta}-i\bar{ \beta}=0$ and $\frac{i}{ \alpha }-i\bar{ \alpha }=0$ which only has the solutions $|\alpha|=| \beta|=1$. The diagonal elements are then all equal to 2 so $\lambda X_2$ will be unitary as long as $|\lambda|^2=\frac{1}{2}$.
\end{proof}

\begin{proposition}
$\lambda X_3$ is unitary when $| \alpha |=1$, $| \beta|=1$, and $|\lambda|^2 = \frac{1}{2}$
\end{proposition}
\begin{proof}
\begin{align*}
\lambda X_3 (\lambda X_3)^{\dagger}=& \lambda \bar{\lambda} \begin{pmatrix}
   \alpha   \beta^2 \bar{ \alpha } \bar{ \beta}^2+1 & 0 & 0 & 0 & 0 & 0 & 0 &  \alpha   \beta^2-\frac{1}{\bar{ \alpha } \bar{ \beta}^2} \\
 0 &  \beta \bar{ \beta}+1 & 0 & 0 & 0 & 0 & \frac{1}{\bar{ \beta}}- \beta & 0 \\
 0 & 0 &  \alpha  \bar{ \alpha }+1 & 0 & 0 & a-\frac{1}{\bar{a}} & 0 & 0 \\
 0 & 0 & 0 & \frac{1}{ \beta \bar{ \beta}}+1 & \bar{ \beta}-\frac{1}{ \beta} & 0 & 0 & 0 \\
 0 & 0 & 0 &  \beta-\frac{1}{\bar{ \beta}} &  \beta \bar{ \beta}+1 & 0 & 0 & 0 \\
 0 & 0 & \bar{ \alpha }-\frac{1}{ \alpha } & 0 & 0 & \frac{1}{ \alpha  \bar{ \alpha }}+1 & 0 & 0 \\
 0 & \frac{1}{ \beta}-\bar{ \beta} & 0 & 0 & 0 & 0 & \frac{1}{ \beta \bar{ \beta}}+1 & 0 \\
 \bar{ \alpha } \bar{ \beta}^2-\frac{1}{ \alpha   \beta^2} & 0 & 0 & 0 & 0 & 0 & 0 & \frac{1}{a  \beta^2 \bar{ \alpha }
   \bar{ \beta}^2}+1 \\
\end{pmatrix}
\end{align*}
To make the off-diagonal elements zero we need $\frac{1}{ \beta}-\bar{ \beta}=0$ and $\frac{1}{ \alpha }-\bar{ \alpha }=0$ which only has the solutions $| \alpha |=| \beta|=1$. The diagonal elements are then all equal to 2 so $\lambda X_3$ will be unitary as long as $|\lambda|^2=\frac{1}{2}$.
\end{proof}

\begin{proposition}
$\lambda X_4$ is unitary when $| \alpha |^2=\frac{\delta-2}{\bar{\delta}}$, $| \beta|^2=\bar{\delta}^2-2\bar{\delta}+2$, $|\lambda|^2 = \frac{1}{1+| \beta|^2}$, and $\delta$ is one of the following: $1+i$, $1-i$, $1$, $\frac{5}{4}+\frac{\sqrt{7}}{4}$, $\frac{5}{4}-\frac{\sqrt{7}}{4}$.
\end{proposition}
\begin{proof}
Let $f=\delta^2-2\delta+2$ we then have:
\begin{align*}
&\lambda X_4 (\lambda X_4)^{\dagger}=\\
&\lambda \bar{\lambda} 
\begingroup
\setlength\arraycolsep{.2pt}
\begin{pmatrix}
 \frac{ \alpha   \beta^2 \bar{ \alpha } \bar{ \beta}^2}{f \bar{f}}+(\delta-2) \left(\bar{\delta}-2\right) & 0 & 0 & 0 & 0 & 0 & 0 & \frac{ \alpha   \beta^2
   \bar{\delta}}{f}-\frac{(\delta-2) \bar{f}}{\bar{ \alpha } \bar{ \beta}^2}
   \\
 0 &  \beta \bar{ \beta}+1 & 0 & 0 & 0 & 0 & \frac{\bar{f}}{\bar{ \beta}}+ \beta & 0 \\
 0 & 0 &  \alpha  \bar{ \alpha }+(\delta-2) \left(\bar{\delta}-2\right) & 0 & 0 & \frac{2-\delta}{\bar{ \alpha }}+ \alpha  \bar{\delta} &
   0 & 0 \\
 0 & 0 & 0 & \frac{f \bar{f}}{ \beta
   \bar{ \beta}}+1 & \bar{ \beta}+\frac{f}{ \beta} & 0 & 0 & 0 \\
 0 & 0 & 0 & \frac{\bar{f}}{\bar{ \beta}}+ \beta &  \beta \bar{ \beta}+1 & 0 & 0 & 0 \\
 0 & 0 & \frac{ \alpha  \delta \bar{a}-\bar{\delta}+2}{a} & 0 & 0 & \frac{1}{ \alpha  \bar{ \alpha }}+\delta \bar{\delta} & 0 & 0
   \\
 0 & \bar{ \beta}+\frac{f}{ \beta} & 0 & 0 & 0 & 0 & \frac{f \bar{f}}{ \beta \bar{ \beta}}+1 & 0 \\
 \frac{\delta \bar{ \alpha } \bar{ \beta}^2}{\bar{f}}-\frac{f
   \bar{f}}{ \alpha   \beta^2} & 0 & 0 & 0 & 0 & 0 & 0 & \frac{f
   \bar{f}}{a  \beta^2 \bar{ \alpha } \bar{ \beta}^2}+\delta \bar{\delta} \\
\end{pmatrix}
\endgroup
\end{align*}
Because we need $\frac{2-\delta}{\bar{ \alpha }} +  \alpha \bar{\delta} =0$ we can set $ \alpha  \bar{ \alpha }=| \alpha |^2=\frac{\delta-2}{\bar{\delta}}$. We also need $\frac{\bar{f}}{\bar{ \beta}}+ \beta=0$ which will happen if $ \beta \bar{ \beta}=| \beta|^2=\bar{f} = \bar{\delta}^2 -2\bar{\delta}+2$. Clearing the denominator and substituting into the upper-rightmost element gives us:

\begin{align*}
 \alpha  \bar{ \alpha } ( \beta \bar{ \beta})^2 \bar{\delta} - (\delta-2)f\bar{f} = 
\frac{(\delta-2)(\bar{\delta}^2-2\bar{\delta}+2)}{\bar{\delta}} - (\delta-2)(\bar{\delta}^2-2\bar{\delta}+2)(\delta^2-2\delta+2)
\end{align*}
Clearing the denominator again this simplifies to:

\begin{align*}
&(\delta-2)(\bar{\delta}^2-2\bar{\delta}+2) - \bar{\delta}(\delta-2)(\bar{\delta}^2-2\bar{\delta}+2)(\delta^2-2\delta+2) = (\delta-2)(\bar{\delta}^2-2\bar{\delta}+2)(1-\bar{\delta}(\delta^2-2\delta+2))
\end{align*}
We can eliminate the case $\delta=2$ since that breaks the $X$-shape. There are then two possibilities, either $\bar{\delta}^2-2\bar{\delta}+2=0$, or $\bar{\delta}(\delta^2-2\delta+2)=1$. In the first case the only two solutions are $\delta=1+i$ or $\delta=1-i$. Solving the second case we substitute $\delta=x+iy$ where $x$ and $y$ are real:
\begin{align*}
\bar{\delta}(\delta^2-2\delta+2)&=(x-iy)((x+iy)^2-2(x+iy)+2)=1\\
\end{align*}
This can then be split into real and imaginary parts:
\begin{align*}
x^3-2x^2+xy^2+2x-2y^2-1&=0\\
y(x^2+y^2-2)=0
\end{align*}
This has the following solutions: $y=0$ and $x=1$ or $x^2+y^2=2$.
In the case $x^2+y^2=2$ we need to solve:
\begin{align*}
x^3-2(x^2+y^2)+xy^2+2x-1&=x^3-2(2)+xy^2+2x-1=0\\
x^2+y^2-2&=0\\
\end{align*}
We can set $y=\pm \sqrt{\frac{5-x^3-2x}{x}}$ as long as $x\neq 0$. In the case that $x=0$ we must have $y=\pm \sqrt{2}$ however this does not solve the first equation. Substituting $y=\pm \sqrt{\frac{5-x^3-2x}{x}}$ into the second equation we get:
\begin{align*}
x^2+ \frac{5-x^3-2x}{x}&=2\\
\end{align*}
Multiplying through by $x$ and then subtracting $2x$ from both sides gives us:
\begin{align*}
x^3+5-x^3-2x-2x&=5-4x=0 \\
\end{align*}
Therefore $x=\frac{5}{4}$ and $y=\pm \frac{\sqrt{7}}{4}$.
\end{proof}

\bibliographystyle{plain}

\appendix

\section{Proof of symmetries}

\begin{proof}
Let $R$ be a solution to the $(d,m,l)$-gYBE and let $\lambda$ be a nonzero scalar, and $Q$ a non-singular $d\times d$ matrix.
\begin{enumerate}
\item \begin{align*}
    (\lambda R \otimes I_V^{\otimes l})(I_V^{\otimes l} \otimes \lambda R)(\lambda R\otimes I_V^{\otimes l})&=(I_V^{\otimes l} \otimes \lambda R)(\lambda R \otimes I_V^{\otimes l})(I_V^{\otimes l} \otimes \lambda R)\\
   \lambda^3 (R \otimes I_V^{\otimes l})(I_V^{\otimes l} \otimes R)(R\otimes I_V^{\otimes l})&=\lambda^3(I_V^{\otimes l} \otimes  R)( R \otimes I_V^{\otimes l})(I_V^{\otimes l} \otimes R)\\
   (R \otimes I_V^{\otimes l})(I_V^{\otimes l} \otimes R)(R\otimes I_V^{\otimes l})&=(I_V^{\otimes l} \otimes  R)( R \otimes I_V^{\otimes l})(I_V^{\otimes l} \otimes R)\\
\end{align*}
\item \begin{align*}
    (R^{-1} \otimes I_V^{\otimes l})(I_V^{\otimes l} \otimes R^{-1})(R^{-1} \otimes I_V^{\otimes l})&=(I_V^{\otimes l} \otimes R^{-1})(R^{-1} \otimes I_V^{\otimes l})(I_V^{\otimes l} \otimes R^{-1}) \\
    (R \otimes I_V^{\otimes l})^{-1}(I_V^{\otimes l} \otimes R)^{-1}(R\otimes I_V^{\otimes l})^{-1}&=(I_V^{\otimes l} \otimes  R)^{-1}( R \otimes I_V^{\otimes l})^{-1}(I_V^{\otimes l} \otimes R)^{-1} \\
   ( (R \otimes I_V^{\otimes l})(I_V^{\otimes l} \otimes R)(R\otimes I_V^{\otimes l}))^{-1} &=( (I_V^{\otimes l} \otimes  R)( R \otimes I_V^{\otimes l})(I_V^{\otimes l} \otimes R))^{-1}\\
   (R \otimes I_V^{\otimes l})(I_V^{\otimes l} \otimes R)(R\otimes I_V^{\otimes l})&=(I_V^{\otimes l} \otimes  R)( R \otimes I_V^{\otimes l})(I_V^{\otimes l} \otimes R)
\end{align*}

\item The same as 2 with $R^{-1}$ replaced by $R^*$
\item The same as 2 with $R^{-1}$ replaced by $R^T$
\item To limit the need for parenthesis denote $\bm{Q} = Q^{-1}$ and consider the left hand side of the gYBE: \begin{align}
\label{QYBE}
       &(Q^{\otimes m}R\bm{Q}^{\otimes m}\otimes I_V^{\otimes l})(I_V^{\otimes l} \otimes Q^{\otimes m}R\bm{Q}^{\otimes m})(Q^{\otimes m}R\bm{Q}^{\otimes m}\otimes I_V^{\otimes l})\\
       =&(Q^{\otimes m} R\otimes I_V^{\otimes l})(\bm{Q}^{\otimes m}\otimes I_V^{\otimes l})(I_V^{\otimes l} \otimes Q^{\otimes m})(I_V^{\otimes l} \otimes R\bm{Q}^{\otimes m})(Q^{\otimes m}R\bm{Q}^{\otimes m}\otimes I_V^{\otimes l})\\
\end{align}
To simplify further we look at the term $(\bm{Q}^{\otimes m}\otimes I_V^{\otimes l})(I_V^{\otimes l} \otimes Q^{\otimes m})$ there are then two cases to consider. In the \textbf{first case} $m>l$ and in the \textbf{second case $m\leq l$}. In the \textbf{first case} we can write:

\begin{align*}
(\bm{Q}^{\otimes m}\otimes I_V^{\otimes l})(I_V^{\otimes l} \otimes Q^{\otimes m})&=(\bm{Q}^{\otimes l}\otimes I_V^{\otimes m-l} \otimes Q^{\otimes l})
\end{align*}

And in the \textbf{second case} we can write:
\begin{align*}
(\bm{Q}^{\otimes m}\otimes I_V^{\otimes l})(I_V^{\otimes l} \otimes Q^{\otimes m})&=(\bm{Q}^{\otimes m}\otimes I_V^{\otimes l-m} \otimes Q^{\otimes m})
\end{align*}

Substituting the \textbf{first case} into equation \ref{QYBE} above we get:
\begin{align*}
&(Q^{\otimes m} R\otimes I_V^{\otimes l})(\bm{Q}^{\otimes l}\otimes I_V^{\otimes m-l} \otimes Q^{\otimes l})(I_V^{\otimes l} \otimes R\bm{Q}^{\otimes m})(Q^{\otimes m}R\bm{Q}^{\otimes m}\otimes I_V^{\otimes l})\\
=&(Q^{\otimes m} R\otimes I_V^{\otimes l})(\bm{Q}^{\otimes l}\otimes (I_V^{\otimes m-l} \otimes Q^{\otimes l}) R\bm{Q}^{\otimes m})(Q^{\otimes m} \otimes I_V^{\otimes l})(R\bm{Q}^{\otimes m}\otimes I_V^{\otimes l})\\
=&(Q^{\otimes m} R\otimes I_V^{\otimes l})
(I_V^{\otimes l}\otimes (I_V^{\otimes m-l} \otimes Q^{\otimes l}) R\bm{Q}^{\otimes m}(Q^{\otimes m-l} \otimes I_V^{\otimes l}))
(R\bm{Q}^{\otimes m}\otimes I_V^{\otimes l})\\
=&(Q^{\otimes m} R\otimes I_V^{\otimes l})
(I_V^{\otimes l}\otimes (I_V^{\otimes m-l} \otimes Q^{\otimes l}) R(I_V^{\otimes m-l} \otimes \bm{Q}^{\otimes l}))
(R\bm{Q}^{\otimes m}\otimes I_V^{\otimes l})\\
=&(Q^{\otimes m} R\otimes I_V^{\otimes l})(I_V^{\otimes l}\otimes I_V^{\otimes m-l}\otimes Q^{\otimes l})(I_V^{\otimes l}\otimes R)(I_V^{\otimes l} \otimes I_V^{\otimes m-l} \otimes \bm{Q}^{\otimes l})
(R\bm{Q}^{\otimes m}\otimes I_V^{\otimes l})\\
=&(Q^{\otimes m} R\otimes I_V^{\otimes l})(I_V^{\otimes m}\otimes Q^{\otimes l})(I_V^{\otimes l}\otimes R)(I_V^{\otimes m} \otimes \bm{Q}^{\otimes l})
(R\bm{Q}^{\otimes m}\otimes I_V^{\otimes l})\\
=&(Q^{\otimes m} R\otimes Q^{\otimes l})(I_V^{\otimes l}\otimes R)
(R\bm{Q}^{\otimes m}\otimes \bm{Q}^{\otimes l})\\
=&(Q^{\otimes m} \otimes  Q^{\otimes l})(R \otimes I_V^{\otimes l})(I_V^{\otimes l}\otimes R)
(R \otimes I_V^{\otimes l})(\bm{Q}^{\otimes m}\otimes \bm{Q}^{\otimes l})\\
=&Q^{\otimes m+l}(R \otimes I_V^{\otimes l})(I_V^{\otimes l}\otimes R)
(R \otimes I_V^{\otimes l})\bm{Q}^{\otimes m+l}\\
\end{align*}

By a similar argument the right hand side of the gYBE can be simplified to:
\begin{align*}
& Q^{\otimes m+l} (I_V^{\otimes l} \otimes R)(R \otimes I_V^{\otimes l})(I_V^{\otimes l} \otimes R) \bm{Q}^{\otimes m+l}\\
\end{align*}

Conjugating by $Q^{\otimes m}$ we get the gYBE. We can similarly handle the \textbf{second case} when $m\leq l$ by substituting into equation \ref{QYBE}:
\begin{align*}
&(Q^{\otimes m} R\otimes I_V^{\otimes l})(\bm{Q}^{\otimes m}\otimes I_V^{\otimes l-m} \otimes Q^{\otimes m})(I_V^{\otimes l} \otimes R\bm{Q}^{\otimes m})(Q^{\otimes m}R\bm{Q}^{\otimes m}\otimes I_V^{\otimes l})\\
=&(Q^{\otimes m} R\otimes I_V^{\otimes l})(\bm{Q}^{\otimes m}\otimes I_V^{\otimes l-m} \otimes Q^{\otimes m} R\bm{Q}^{\otimes m})(Q^{\otimes m} \otimes I_V^{\otimes l})(R\bm{Q}^{\otimes m}\otimes I_V^{\otimes l})\\
=&(Q^{\otimes m} R\otimes I_V^{\otimes l})(I_V^{\otimes m}\otimes I_V^{\otimes l-m} \otimes Q^{\otimes m} R\bm{Q}^{\otimes m})(R\bm{Q}^{\otimes m}\otimes I_V^{\otimes l})\\
=&(Q^{\otimes m} R\otimes I_V^{\otimes l})
(I_V^{\otimes m}\otimes I_V^{\otimes l-m} \otimes Q^{\otimes m})(I_V^{\otimes m}\otimes I_V^{\otimes l-m} R)
(I_V^{\otimes m}\otimes I_V^{\otimes l-m} \otimes \bm{Q}^{\otimes m})(R\bm{Q}^{\otimes m}\otimes I_V^{\otimes l})\\
=&(Q^{\otimes m} R\otimes I_V^{\otimes l-m} \otimes Q^{\otimes m})
(I_V^{\otimes l}\otimes R) (R\bm{Q}^{\otimes m}\otimes I_V^{\otimes l-m} \otimes \bm{Q}^{\otimes m})\\
=&(Q^{\otimes m} \otimes I_V^{\otimes l-m} \otimes Q^{\otimes m})(R \otimes I_V^{\otimes l})
(I_V^{\otimes l}\otimes R) (R\otimes I_V^{\otimes l}) (\bm{Q}^{\otimes m}\otimes I_V^{\otimes l-m} \otimes \bm{Q}^{\otimes m})\\
\end{align*}
The right hand side of the gYBE can similarly be manipulated to:
\begin{align*}
&(Q^{\otimes m} \otimes I_V^{\otimes l-m} \otimes Q^{\otimes m})(I_V^{\otimes l} \otimes R)(R \otimes I_V^{\otimes l})(I_V^{\otimes l} \otimes R) (\bm{Q}^{\otimes m}\otimes I_V^{\otimes l-m} \otimes \bm{Q}^{\otimes m}) \\
\end{align*}
Therefore if $R$ is a solution to the gYBE so is $Q^{\otimes m} R \bm{Q}^{\otimes m}$.
\end{enumerate}
\end{proof}

\section{Equivalent forms of the aYBE}

Here we show how to obtain the form of the aYBE in equation \ref{eq:aYBE} from equation \ref{AlgebraicYBE}. In equation \ref{AlgebraicYBE} $R_{ab}$ acts on the factors $a$ and $b$ and does not affect the third factor. For example $R_{13}$ acts on the first and third factors and does not affect the middle factor:
\begin{align*} 
R_{13} (e_i \otimes e_j \otimes e_k) = \sum_{ab} R_{ik}^{ab} (e_a \otimes e_j \otimes e_b)
\end{align*}
The left hand side of equation \ref{AlgebraicYBE} acts on $(e_{j_1} \otimes e_{j_2} \otimes e_{j_3})$ as follows:

\begin{align*}
    R_{12}R_{13}R_{23}(e_{j_1} \otimes e_{j_2} \otimes e_{j_3}) &= R_{12}R_{13} \sum_{k_2,k_3} R_{{j_2} j_3}^{k_2k_3} (e_{j_1} \otimes e_{k_2} \otimes e_{k_3})\\
    &= R_{12} \sum_{k_2, k_3, k_1, l_3} R_{{j_2} j_3}^{k_2k_3} R_{j_1 k_3}^{k_1 l_3} (e_{k_1} \otimes e_{k_2} \otimes e_{l_3})\\
    &= \sum_{k_2, k_3, k_1, l_3, l_1, l_2} R_{{j_2} j_3}^{k_2k_3} R_{j_1 k_3}^{k_1 l_3} R_{k_1 k_2}^{l_1 l_2} (e_{l_1} \otimes e_{l_2} \otimes e_{l_3})\\
\end{align*}
The action of the right hand side of equation \ref{AlgebraicYBE} can similarly be written:
\begin{align*}
    R_{23}R_{13}R_{12}(e_{j_1} \otimes e_{j_2} \otimes e_{j_3}) &= R_{23}R_{13} \sum_{k_1, k_2} R_{j_1 j_2}^{k_1 k_2} (e_{k_1} \otimes e_{k_2} \otimes e_{j_3})\\
     &= R_{23} \sum_{k_1,k_2,l_1,k_3} R_{ j_2 j_2}^{k_1 k_2} R_{k_1 j_3}^{l_1 k_3} (e_{l_1} \otimes e_{k_2} \otimes e_{k_3})\\
     &= \sum_{k_1,k_2,l_1,k_3,l_2, l_3} R_{j_1 j_2}^{k_1 k_2} R_{k_1 j_3}^{l_1 k_3} R_{k_2 k_3}^{l_2 l_3} (e_{l_1} \otimes e_{l_2} \otimes e_{l_3})\\
\end{align*}
Therefore equation \ref{AlgebraicYBE} is equivalent to equation \ref{eq:aYBE}.

\bibliography{Univbib}

\end{document}